\documentclass{llncs}
\usepackage{version}
\usepackage{pasi}

\title{Process Algebra with Strategic Interleaving, \\ Revised Version}
\author{J.A. Bergstra \and C.A. Middelburg}
\institute{Informatics Institute, Faculty of Science, University of
           Amsterdam, \\
           Science Park~904, 1098~XH Amsterdam, the Netherlands \\
           \email{J.A.Bergstra@uva.nl,C.A.Middelburg@uva.nl}}

\begin{document}
\maketitle

\begin{abstract}
In process algebras such as \ACP\ (Algebra of Communicating Processes), 
parallel processes are considered to be interleaved in an arbitrary way.
In the case of multi-threading as found in contemporary programming 
languages, parallel processes are actually interleaved according to some 
interleaving strategy.
An interleaving strategy is what is called a process-scheduling policy 
in the field of operating systems. 
In many systems, for instance hardware/software systems, we have to do 
with both parallel processes that may best be considered to be 
interleaved in an arbitrary way and parallel processes that may best be 
considered to be interleaved according to some interleaving strategy.
Therefore, we extend \ACP\ in this paper with the latter form of 
interleaving.
The established properties of the extension concerned include an
elimination property, a conservative extension property, and a unique 
expansion property.
\begin{classcode}
D.1.3, D.4.1, F.1.2
\end{classcode}
\par\addvspace{1.5ex}
{\sl Note:} 
In the published version of ``Process algebra with strategic 
interleaving'' (Theory of Computing Systems 63(3), 488--505 (2019)) the 
proof outline of Theorem~\ref{theorem-reduction} appears to be 
inadequate.
In particular, a head normal form result is used that is too weak. 
This revision of the published version has been made to rectify 
this.
\end{abstract}

\section{Introduction}
\label{sect-intro}

In algebraic theories of processes, such as \ACP~\cite{BW90,BK84b}, 
CCS~\cite{HM85,Mil89} and CSP~\cite{BHR84,Hoa85}, processes are discrete 
behaviours that proceed by doing steps in a sequential fashion.
The parallel composition of two processes is usually considered to 
incorporate all conceivable interleavings of their steps. 
In each interleaving, the steps of both processes occur in some order 
where each time one step is taken from either of the processes.
According to many, this interpretation of parallel composition, called 
arbitrary interleaving, is a plausible, general, if not idealized 
interpretation. 
Underlying the usual justification of this claim is the assumption that 
at most one step is done at each point in time.
However, others contend that interpretations in which this simplifying 
assumption is fulfilled are not faithful.
Be that as it may, arbitrary interleaving turns out to be appropriate 
for many applications and to facilitate formal algebraic reasoning. 

Multi-threading as found in programming languages such as 
Java~\cite{GJSB00a} and C\#~\cite{HWG03a}, gives rise to parallel 
composition of processes.
In the case of multi-threading, however, the steps of the processes 
concerned are interleaved according to a process-scheduling policy.
We use the term strategic interleaving for this more constrained form of
interleaving; and we further use the term interleaving strategy instead 
of process-scheduling policy.
Arbitrary interleaving and strategic interleaving are quite different.
The following points illustrate this: (a) whether the interleaving of 
certain processes leads to inactiveness depends on the interleaving 
strategy used; (b) sometimes inactiveness occurs with a particular 
interleaving strategy whereas arbitrary interleaving would not lead to 
inactiveness and vice versa.

In previous work, we studied strategic interleaving in the setting of 
thread algebra, which is built on a specialized algebraic theory of 
processes devoted to the behaviours produced by instruction sequences 
under execution (see e.g.~\cite{BM04c,BM06a,BM07a}).
We have, for instance, given demonstrations of points (a) and (b) above 
in this setting. 
Nowadays, multi-threading is often used in the implementation of 
systems.
Because of this, in many systems, for instance hardware/software 
systems, we have to do with parallel processes that may best be 
considered to be interleaved in an arbitrary way as well as parallel 
processes that may best be considered to be interleaved according to 
some interleaving strategy.
This is what motivated us to do the work presented in this paper, 
namely extending \ACP\ such that it supports both arbitrary interleaving
and strategic interleaving.

To our knowledge, there exists no work on strategic interleaving in the 
setting of a general algebraic theory of processes like ACP, CCS and 
CSP.
In the work presented in this paper, we consider strategic interleaving
where process creation is taken into account.
The approach to process creation followed in this paper originates from
the one first followed in~\cite{Ber90a} to extend \ACP\ with process 
creation and later followed \linebreak[2] in~\cite{BB93a,BMU98a,BM02a} 
to extend different timed versions of \ACP\ with process creation.
The only other approach that we know of is the approach, based 
on~\cite{AB88a}, that has for instance been followed 
in~\cite{BV92a,GR97a}.
However, with that approach, it is unlikely that data about the
creation of processes can be made available for the decision making 
con\-cerning the strategic interleaving of processes.

The extension of \ACP\ presented in this paper covers a generic
interleaving strategy that can be instantiated with different specific
interleaving strategies.
We found two plausible ways to deal with inactiveness of a process whose 
steps are being interleaved with steps of other processes in the case of 
strategic interleaving.
This gives rise to two plausible extensions of \ACP.
We will treat only one of them in detail.

The rest of this paper is organized as follows.
In Section~\ref{sect-ACPrec}, we review \ACP\ 
(Section~\ref{subsect-ACP}) and guarded recursion in the setting of 
\ACP\ (Section~\ref{subsect-REC}).
In Section~\ref{sect-SI}, we extend \ACP\ with strategic interleaving
(Section~\ref{subsect-SI}) and establish some important properties of 
the extension (Section~\ref{subsect-theorems-siACP}).  
In Section~\ref{sect-concl}, we make some concluding remarks.

\section{ACP with Guarded Recursion}
\label{sect-ACPrec}

In this section, we give a survey of \ACP\ (Algebra of Communicating 
Processes) and guarded recursion in the setting of \ACP.
For a comprehensive overview, the reader is referred 
to~\cite{BW90,Fok00}.

\subsection{\ACP}
\label{subsect-ACP}

In \ACP, it is assumed that a fixed but arbitrary set $\Act$ of
\emph{actions}, with $\dead \notin \Act$, has been given.
We write $\Actd$ for $\Act \union \set{\dead}$.
It is further assumed that a fixed but arbitrary commutative and 
associative \emph{communication} function 
$\funct{\commf}{\Actd \x \Actd}{\Actd}$, with 
$\commf(\dead,a) = \dead$ for all $a \in \Actd$, has been given.
The function $\commf$ is regarded to give the result of synchronously
performing any two actions for which this is possible, and to give 
$\dead$ otherwise.

The signature of \ACP\ consists of the following constants and 
operators:
\begin{itemize}
\item
for each $a \in \Act$, the \emph{action} constant 
$a$\,;
\item
the \emph{inaction} constant $\dead$\,;
\item
the binary \emph{alternative composition} operator 
$\ph \altc \ph$\,;
\item
the binary \emph{sequential composition} operator 
$\ph \seqc \ph$\,;
\item
the binary \emph{parallel composition} operator 
$\ph \parc \ph$\,;
\item
the binary \emph{left merge} operator 
$\ph \leftm \ph$\,;
\item
the binary \emph{communication merge} operator 
$\ph \commm \ph$\,;
\item
for each $H \subseteq \Act$, the unary \emph{encapsulation} operator
$\encap{H}$\,.
\end{itemize}
We assume that there are infinitely many variables, including $x,y,z$.
Terms are built as usual.
We use infix notation for the binary operators.
The precedence conventions used with respect to the operators of \ACP\
are as follows: $\altc$ binds weaker than all others, $\seqc$ binds
stronger than all others, and the remaining operators bind equally
strong.

The constants and operators of \ACP\ can be explained as follows:
\begin{itemize}
\item
the constant $a$ denotes the process that is only capable of first 
performing action $a$ and next terminating successfully;
\item
the constant $\dead$ denotes the process that is not capable of doing 
anything;
\item
a closed term of the form $t \altc t'$ denotes the process that behaves 
either as the process denoted by $t$ or as the process denoted by $t'$, 
but not both;
\item
a closed term of the form $t \seqc t'$ denotes the process that first 
behaves as the process denoted by $t$ and on successful termination of 
that process it next behaves as the process denoted by $t'$;
\item
a closed term of the form $t \parc t'$ denotes the process that behaves 
as the process that proceeds with the processes denoted by $t$ and $t'$ 
in parallel;
\item
a closed term of the form $t \leftm t'$ denotes the process that behaves 
the same as the process denoted by $t \parc t'$, except that it starts 
with performing an action of the process denoted by $t$;
\item
a closed term of the form $t \commm t'$ denotes the process that behaves 
the same as the process denoted by $t \parc t'$, except that it starts 
with performing an action of the process denoted by $t$ and an action of 
the process denoted by~$t'$ synchronously;
\item
a closed term of the form $\encap{H}(t)$ denotes the process that 
behaves the same as the process denoted by $t$, except that actions from 
$H$ are blocked.
\end{itemize}
The operators $\leftm$ and $\commm$ are of an auxiliary nature.
They are needed to axiomatize \ACP.

The axioms of \ACP\ are the equations given in Table~\ref{axioms-ACP}.
\begin{table}[!t]
\caption{Axioms of \ACP}
\label{axioms-ACP}
\begin{eqntbl}
\begin{axcol}
x \altc y = y \altc x                                  & \axiom{A1}   \\
(x \altc y) \altc z = x \altc (y \altc z)              & \axiom{A2}   \\
x \altc x = x                                          & \axiom{A3}   \\
(x \altc y) \seqc z = x \seqc z \altc y \seqc z        & \axiom{A4}   \\
(x \seqc y) \seqc z = x \seqc (y \seqc z)              & \axiom{A5}   \\
x \altc \dead = x                                      & \axiom{A6}   \\
\dead \seqc x = \dead                                  & \axiom{A7}   \\
{}                                                                    \\
\encap{H}(a) = a                \hfill \mif a \notin H & \axiom{D1}   \\
\encap{H}(a) = \dead            \hfill \mif a \in H    & \axiom{D2}   \\
\encap{H}(x \altc y) = \encap{H}(x) \altc \encap{H}(y) & \axiom{D3}   \\
\encap{H}(x \seqc y) = \encap{H}(x) \seqc \encap{H}(y) & \axiom{D4}  
\end{axcol}
\qquad \qquad
\begin{axcol}
x \parc y =
          x \leftm y \altc y \leftm x \altc x \commm y & \axiom{CM1}  \\
a \leftm x = a \seqc x                                 & \axiom{CM2}  \\
a \seqc x \leftm y = a \seqc (x \parc y)               & \axiom{CM3}  \\
(x \altc y) \leftm z = x \leftm z \altc y \leftm z     & \axiom{CM4}  \\
a \seqc x \commm b = \commf(a,b) \seqc x               & \axiom{CM5}  \\
a \commm b \seqc x = \commf(a,b) \seqc x               & \axiom{CM6}  \\
a \seqc x \commm b \seqc y =
                        \commf(a,b) \seqc (x \parc y)  & \axiom{CM7}  \\
(x \altc y) \commm z = x \commm z \altc y \commm z     & \axiom{CM8}  \\
x \commm (y \altc z) = x \commm y \altc x \commm z     & \axiom{CM9}  \\
\dead \commm x = \dead                                 & \axiom{CM10} \\
x \commm \dead = \dead                                 & \axiom{CM11} \\
a \commm b = \commf(a,b)                               & \axiom{CM12}  
\end{axcol}
\end{eqntbl}
\end{table}
In these equations, $a$, $b$ and $c$ stand for arbitrary constants of 
\ACP, and $H$ stands for an arbitrary subset of $\Act$.
Moreover, $\commf(a,b)$ stands for the action constant for the action 
$\commf(a,b)$.
In D1 and D2, side conditions restrict what $a$ and $H$ stand for.

In other presentations of \ACP, $\commf(a,b)$ is regularly replaced by 
$a \commm b$ in CM5--CM7.
By CM12, which is more often called CF, these replacements give rise to
an equivalent axiomatization.
In other presentations of \ACP, CM10 and CM11 are usually absent.
These equations are not derivable from the other axioms, but all there 
closed substitution instances are derivable from the other axioms.
Moreover, CM10 and CM11 hold in virtually all models of ACP that have 
been devised.

In the sequel, we will use the sum notation $\Altc{i<n} t_i$.
For each $i \in \Nat$, let $t_i$ be a term of \ACP\ or an extension of 
\ACP.
Then $\Altc{i<0} t_i = \dead$ and, for each $n \in \Natpos$,%
\footnote
{We write $\Natpos$ for the set $\set{n \in \Nat \where n \geq 1}$ of 
 positive natural numbers.} 
the term $\Altc{i<n} t_i$ is defined by induction on $n$ as follows:
$\Altc{i<1} t_i = t_0$ and $\Altc{i<n+1} t_i =\Altc{i<n} t_i \altc t_n$.

\subsection{Guarded Recursion}
\label{subsect-REC}

A closed \ACP\ term denotes a process with a finite upper bound to the 
number of actions that it can perform. 
Guarded recursion allows the description of processes without a finite 
upper bound to the number of actions that it can perform.

Let $T$ be \ACP\ or a concrete extensions of \ACP,%
\footnote
{A concrete extension of \ACP\ is an extension of \ACP\ that does not
 offer the possibility of abstraction from certain actions.
 All extensions of \ACP\ introduced in this paper are concrete
 extensions.} 
and let $t$ be a $T$ term containing a variable $X$.
Then an occurrence of $X$ in $t$ is \emph{guarded} if $t$ has a subterm 
of the form $a \seqc t'$ where $a \in \Act$ and $t'$ is a $T$ term 
containing this occurrence of $X$.

Let $T$ be \ACP\ or a concrete extension of \ACP.
Then a $T$ term $t$ is a \emph{guarded} $T$ term if all occurrences of 
variables in $t$ are guarded.

Let $T$ be \ACP\ or a concrete extension of \ACP.
Then a \emph{guarded recursive specification} over $T$ is a finite or
countably infinite set of recursion equations 
$E = \set{X = t_X \where X \in V}$, 
where $V$ is a set of variables and each $t_X$ is either a guarded $T$ 
term in which variables other than the variables from $V$ do not occur 
or a $T$ term rewritable to such a term using the axioms of $T$ in 
either direction and/or the equations in $E$, except the equation 
$X = t_X$, from left to right.
We write $\vars(E)$ for the set of all variables that occur in $E$.
A solution of $E$ in some model of $T$ is a set 
$\set{P_X \where X \in \vars(E)}$ of elements of the carrier of that 
model such that the equations of $E$ hold if, for all $X \in \vars(E)$, 
$X$ is assigned $P_X$.
We are only interested models of \ACP\ and concrete extensions of \ACP\ 
in which guarded recursive specifications have unique solutions. 

Let $T$ be \ACP\ or a concrete extension of \ACP.
We extend $T$ with guarded recursion by adding constants for solutions 
of guarded recursive specifications over $T$ and axioms concerning these 
additional constants.
For each guarded recursive specification $E$ over $T$ and each 
$X \in \vars(E)$, we add a constant standing for the unique solution of 
$E$ for $X$ to the constants of $T$.
The constant standing for the unique solution of $E$ for $X$ is denoted 
by $\rec{X}{E}$.
We use the following notation.
Let $t$ be a $T$ term and $E$ be a guarded recursive specification.
Then we write $\rec{t}{E}$ for $t$ with, for all $X \in \vars(E)$, all
occurrences of $X$ in $t$ replaced by $\rec{X}{E}$.
We add the equation RDP and the conditional equation RSP given in 
Table~\ref{axioms-REC} to the axioms of $T$.
\begin{table}[!t]
\caption{Axioms for guarded recursion}
\label{axioms-REC}
\begin{eqntbl}
\begin{saxcol}
\rec{X}{E} = \rec{t_X}{E} & \mif X \!=\! t_X \in E     & \axiom{RDP} \\
E \Limpl X = \rec{X}{E}   & \mif X \in \vars(E)        & \axiom{RSP} 
\end{saxcol}
\end{eqntbl}
\end{table}
In RDP and RSP, $X$ stands for an arbitrary variable, $t_X$ stands for 
an arbitrary $T$ term, and $E$ stands for an arbitrary guarded recursive 
specification over $T$.
Side conditions restrict what $X$, $t_X$ and $E$ stand for.
We write $T_\mathrm{rec}$ for the resulting theory.

The equations $\rec{X}{E} = \rec{t_X}{E}$ for a fixed $E$ express that 
the constants $\rec{X}{E}$ make up a solution of $E$. 
The conditional equations $E \Limpl X = \rec{X}{E}$ express that this 
solution is the only one.

In extensions of \ACP\ whose axioms include RSP, we have to deal with 
conditional equational formulas with a countably infinite number of 
premises.
Therefore, infinitary conditional equational logic is used in deriving 
equations from the axioms of extensions of \ACP\ whose axioms include 
RSP.
A complete inference system for infinitary conditional equational logic 
can be found in, for example,~\cite{GV93}.
In the case of infinitary conditional equational logic derivation trees 
may be infinitely branching, but they may not have infinite branches.

\section{Strategic Interleaving}
\label{sect-SI}

In this section, we extend \ACP\ with strategic interleaving, i.e.\
interleaving according to some interleaving strategy.
Interleaving strategies are abstractions of scheduling algorithms.
Interleaving according to some interleaving strategy is what really 
happens in the case of multi-threading as found in contemporary 
programming languages.

\subsection{\ACP\ with Strategic Interleaving}
\label{subsect-SI}

In the extension of \ACP\ with strategic interleaving presented below, 
it is expected that an interleaving strategy uses the interleaving 
history in one way or another to make process-scheduling decisions.

The set $\Hist$ of \emph{interleaving histories} is the subset of 
$\seqof{(\Natpos \x \Natpos)}$ that is inductively defined by the 
following rules:
\begin{itemize}
\item
$\emptyseq \in \Hist$;
\item
if $i \leq n$, then $\tup{i,n} \in \Hist$;
\item
if $h \concat \tup{i,n} \in \Hist$, $j \leq n$, and 
$n - 1 \leq m \leq n + 1$, then 
$h \concat \tup{i,n} \concat \tup{j,m} \in \Hist$.%
\footnote
{We write
 $\emptyseq$ for the empty sequence, 
 $d$ for the sequence having $d$ as sole element, and 
 $\alpha \concat \alpha'$ for the concatenation of sequences $\alpha$
 and $\alpha'$.
 We assume that the usual identities, such as
 $\emptyseq \concat \alpha = \alpha$ and
 $(\alpha \concat \alpha') \concat \alpha'' =
  \alpha \concat (\alpha' \concat \alpha'')$, hold.}
\end{itemize}
The intuition concerning interleaving histories is as follows:
if the $k$th pair of an interleaving history is $\tup{i,n}$, then the 
$i$th process got a turn in the $k$th interleaving step and after its
turn there were $n$ processes to be interleaved.
The number of processes to be interleaved may increase due to process
creation (introduced below) and decrease due to successful termination 
of processes.
 
The presented extension of \ACP\ is called \siACP\ (\ACP\ with 
Strategic Interleaving). 
It covers a generic interleaving strategy that can be instantiated with 
different specific interleaving strategies that can be represented in 
the way that is explained below.

In \siACP, it is assumed that the following has been given:
\begin{itemize}
\item
a fixed but arbitrary set $S$; 
\item
for each $n \in \Natpos$, a fixed but arbitrary function 
$\funct{\sched{n}}{\Hist \x S}{\set{1,\ldots,n}}$;
\item
for each $n \in \Natpos$, a fixed but arbitrary function 
$\funct{\updat{n}}{\Hist \x S \x \set{1,\ldots,n} \x \Act}{S}$.
\end{itemize}
The elements of $S$ are called \emph{control states}, $\sched{n}$ is 
called an \emph{abstract scheduler} (\emph{for $n$ processes}), and 
$\updat{n}$ is called a \emph{control state transformer} (\emph{for 
$n$ processes}).
The intuition concerning $S$, $\sched{n}$, and $\updat{n}$ is as 
follows:
\begin{itemize}
\item
the control states from $S$ encode data that are relevant to the 
interleaving strategy, but not derivable from the interleaving history;
\item
if $\sched{n}(h,s) = i$, then the $i$th process gets the next turn after 
interleaving history $h$ in control state $s$;
\item
if $\updat{n}(h,s,i,a) = s'$, then $s'$ is the control state that arises 
from the $i$th process doing $a$ after interleaving history $h$ in 
control state $s$.
\end{itemize}
Thus, $S$, $\indfam{\sched{n}}{n \in \Natpos}$, and 
$\indfam{\updat{n}}{n \in \Natpos}$ make up a way to represent an 
interleaving strategy.
This way to represent an interleaving strategy is engrafted 
on~\cite{SS00a}.

Consider the case where $S$ is a singleton set, 
for each $n \in \Natpos$, $\sched{n}$ is defined by
\begin{ldispl}
\sched{n}(\emptyseq,s) = 1\;,  \\
\sched{n}(h \concat \tup{j,n},s) = (j + 1) \bmod n\;,
\end{ldispl}%
and, for each $n \in \Natpos$, $\updat{n}$ is defined by 
\begin{ldispl}
\updat{n}(h,s,i,a) = s\;.
\end{ldispl}%
In this case, the interleaving strategy corresponds to the round-robin 
scheduling algorithm.
More advanced strategies can be obtained if the scheduling makes more 
advanced use of the interleaving history and the control state.
The interleaving history may, for example, be used to factor the 
individual lifetimes of the processes to be interleaved and their 
creation hierarchy into the process-scheduling decision making.
Individual properties of the processes to be interleaved that depend on 
the actions performed by them can be taken into account by making use of 
the control state.
The control state may, for example, be used to factor the processes 
being interleaved that currently wait to acquire a lock from a process 
that manages a shared resource into the process-scheduling decision 
making.%
\footnote
{In~\cite{BM04c}, various examples of interleaving strategies are given 
 in the setting of the relatively unknown thread algebra. 
 The representation of the more serious of these examples in the current 
 setting demands nontrivial use of the control state.}

In \siACP, it is also assumed that a fixed but arbitrary set $D$ of \emph{data} 
and a fixed but arbitrary function $\funct{\crea}{D}{P}$, where $P$ is 
the set of all closed terms over the signature of \siACP\ (given below), 
have been given and that, for each $d \in D$ and $a, b \in \Act$, 
$\pcr(d),\rcr(d) \in \Act$, $\commf(\pcr(d),a) = \dead$, and
$\commf(a,b) \neq \pcr(d)$.
The action $\pcr(d)$ can be considered a process creation request and 
the action $\rcr(d)$ can be considered a process creation act.
They represent the request to start the process denoted by $\crea(d)$ in 
parallel with the requesting process and the act of carrying out that 
request, respectively.

The signature of \siACP\ consists of the constants and operators
from the signature of \ACP\ and in addition the following operators:
\begin{itemize}
\item
for each $n \in \Natpos$, $h \in \Hist$, and $s \in S$,
the $n$-ary \emph{strategic interleaving} operator $\siop{n}{h}{s}$;
\item
for each $n,i \in \Natpos$ with $i \leq n$, $h \in \Hist$, and 
$s \in S$,
the $n$-ary \emph{positional strategic interleaving} operator
$\posmop{n}{i}{h}{s}$.
\end{itemize}

The strategic interleaving operators can be explained as follows:
\begin{itemize}
\item
a closed term of the form $\si{n}{h}{s}{t_1,\ldots,t_n}$ denotes the 
process that results from interleaving of the $n$ processes denoted by 
$t_1,\ldots,t_n$ after interleaving history $h$ in control state $s$, 
according to the interleaving strategy represented by $S$, 
$\indfam{\sched{n}}{n \in \Natpos}$, and 
$\indfam{\updat{n}}{n \in \Natpos}$.
\end{itemize}
The positional strategic interleaving operators are auxiliary operators 
used to axiomatize the strategic interleaving operators.
The role of the positional strategic interleaving operators in the 
axiomatization is similar to the role of the left merge operator found 
in \ACP.

The axioms of \siACP\ are the axioms of \ACP\ and in addition the 
equations given in Table~\ref{axioms-strategic-interleaving}.
\begin{table}[!t]
\caption{Axioms for strategic interleaving}
\label{axioms-strategic-interleaving}
\begin{eqntbl}
\begin{axcol}
\si{n}{h}{s}{x_1,\ldots,x_n} = 
\posm{n}{\sched{n}(h,s)}{h}{s}{x_1,\ldots,x_n}          & \axiom{SI1}  
\eqnsep
\posm{n}{i}{h}{s}{x_1,\ldots,x_{i-1},\dead,x_{i+1},\ldots,x_n} = \dead
                                                        & \axiom{SI2} \\
\posm{1}{i}{h}{s}{a} = a                                & \axiom{SI3} \\
\posm{n+1}{i}{h}{s}{x_1,\ldots,x_{i-1},a,x_{i+1},\ldots,x_{n+1}} =
\\ \qquad
a \seqc
\si{n}{h \concat \tup{i,n}}{\updat{n+1}(h,s,i,a)}
 {x_1,\ldots,x_{i-1},x_{i+1},\ldots,x_{n+1}}            & \axiom{SI4} \\
\posm{n}{i}{h}{s}{x_1,\ldots,x_{i-1},a \seqc x_i',x_{i+1},\ldots,x_n} =
\\ \qquad
a \seqc
\si{n}{h \concat \tup{i,n}}{\updat{n}(h,s,i,a)}
 {x_1,\ldots,x_{i-1},x_i',x_{i+1},\ldots,x_n}           & \axiom{SI5} \\
\posm{n}{i}{h}{s}{x_1,\ldots,x_{i-1},\pcr(d),x_{i+1},\ldots,x_n} =
\\ \qquad
\rcr(d) \seqc
\si{n}{h \concat \tup{i,n}}{\updat{n}(h,s,i,\pcr(d))}
 {x_1,\ldots,x_{i-1},x_{i+1},\ldots,x_n,\crea(d)}       & \axiom{SI6} \\
\posm{n}{i}{h}{s}
 {x_1,\ldots,x_{i-1},\pcr(d) \seqc x_i',x_{i+1},\ldots,x_n} =
\\ \qquad
\rcr(d) \seqc
\si{n+1}{h \concat \tup{i,n+1}}{\updat{n}(h,s,i,\pcr(d))}
 {x_1,\ldots,x_{i-1},x_i',x_{i+1},\ldots,x_n,\crea(d)}  & \axiom{SI7} \\ 
\posm{n}{i}{h}{s}
 {x_1,\ldots,x_{i-1},x_i' \altc x_i'',x_{i+1},\ldots,x_n} =
\\ \qquad
\posm{n}{i}{h}{s}{x_1,\ldots,x_{i-1},x_i',x_{i+1},\ldots,x_n} \altc
\posm{n}{i}{h}{s}{x_1,\ldots,x_{i-1},x_i'',x_{i+1},\ldots,x_n} 
                                                        & \axiom{SI8}
\end{axcol}
\end{eqntbl}
\end{table}
In the additional equations, $n$ and $i$ stand for arbitrary numbers 
from $\Natpos$ with $i \leq n$, $h$ stands for an arbitrary interleaving 
history from $\Hist$, $s$ stands for an arbitrary control state from 
$S$, $a$ stands for an arbitrary action constant that is not of the form 
$\pcr(d)$ or $\rcr(d)$, and $d$ stands for an arbitrary datum $d$ from 
$D$.

Axiom SI2 expresses that, in the event of inactiveness of the process 
whose turn it is, the whole becomes inactive immediately.
A plausible alternative is that, in the event of inactiveness of the 
process whose turn it is, the whole becomes inactive only after all 
other processes have terminated or become inactive.
In that case, the functions 
$\funct{\updat{n}}{\Hist \x S \x \set{1,\ldots,n} \x \Act}{S}$ must be 
extended to functions 
$\funct{\updat{n}}
 {\Hist \x S \x \set{1,\ldots,n} \x (\Act \union \set{\dead})}{S}$
and axiom SI2 must be replaced by the axioms in 
Table~\ref{axioms-alt-deadlock}.
\begin{table}[!t]
\caption{Alternative axioms for SI2}
\label{axioms-alt-deadlock}
\begin{eqntbl}
\begin{axcol}
\posm{1}{i}{h}{s}{\dead} = \dead                       & \axiom{SI2a} \\
\posm{n+1}{i}{h}{s}{x_1,\ldots,x_{i-1},\dead,x_{i+1},\ldots,x_{n+1}} = 
\\ \qquad
\si{n}{h \concat \tup{i,n}}{\updat{n+1}(h,s,i,\dead)}
      {x_1,\ldots,x_{i-1},x_{i+1},\ldots,x_{n+1}} \seqc \dead         
                                                       & \axiom{SI2b} 
\end{axcol}
\end{eqntbl}
\end{table}

In \siACPr, i.e.\ \siACP\ extended with guarded recursion in the way 
described in Section~\ref{sect-ACPrec}, the processes that can be created 
are restricted to the ones denotable by a closed \siACP\ term.
This restriction stems from the requirement that $\crea$ is a function 
from $D$ to the set of all closed \siACP\ terms.
The restriction can be removed by relaxing this requirement to the 
requirement that $\crea$ is a function from $D$ to the set of all closed 
\siACPr\ terms. 
We write \siACPrp\ for the theory resulting from this relaxation.
In other words, \siACPrp\ differs from \siACPr\ in that it is assumed 
that a fixed but arbitrary function $\funct{\crea}{D}{P}$, where $P$ is 
the set of all closed terms over the signature of \siACPr, has been 
given.

It is customary to associate transition systems with closed terms of the 
language of an ACP-like theory of processes by means of structural 
operational semantics and to use this to construct a model in which 
closed terms are identified if their associated transition systems are 
bisimilar.
The structural operational semantics of \ACP\ can be found 
in~\cite{BW90,Fok00}.
The additional transition rules for the strategic interleaving operators 
and the positional strategic interleaving operators are given in 
Appendix~\ref{appendix-SOS}.

\subsection{Basic Properties of \ACP\ with Strategic Interleaving}
\label{subsect-theorems-siACP}

In this section, the subject of concern is the connection between \ACP\ 
and \siACP.
The main results are an elimination result, a conservative extension 
result, and a unique expansion result.
We begin with establishing some results that will be used in the proof 
of those main results.

Let $T$ be \siACP\ or \siACPr.
Then the set $\HNF$ of \emph{head normal forms of} $T$ is inductively 
defined by the following rules:
\begin{itemize}
\item 
$\dead \in \HNF$;
\item 
if $a \in \Act$, then $a \in \HNF$;
\item 
if $a \in \Act$ and $t$ is a $T$ term, then $a \seqc t \in \HNF$;
\item 
if $t,t' \in \HNF$, then $t \altc t '\in \HNF$.
\end{itemize}
Each head normal form of $T$ is derivably equal to a head normal form of 
the form
$\Altc{i<n} a_i \seqc t_i \altc \Altc{j<m} b_i$,
where $n,m \in \Nat$, for all $i \in \Nat$ with $i < n$, $a_i \in \Act$ 
and $t_i$ is a $T$ term, and, for all $j \in \Nat$ with $j < m$, 
$b_j \in \Act$. 

Each guarded \siACPr\ term is derivably equal to a head normal form of 
\siACPr.
\begin{proposition}[Head normal form]
\label{proposition-HNF-siACP}
For each guarded \siACPr\ term $t$, there exists a head normal form $t'$ 
of \siACPr\ such that $t = t'$ is derivable from the axioms of \siACPr.
\end{proposition}
\begin{proof}
First we prove the following weaker result about head normal forms:
\begin{quote}
\emph{For each guarded \siACP\ term $t$, there exists a head normal form 
$t'$ of \siACP\ such that $t = t'$ is derivable from the axioms of 
\siACP}.
\end{quote}
The proof is straightforward by induction on the structure of $t$.
The case where $t$ is of the form $\dead$ and the case where $t$ is of 
the form $a$ ($a \in \Act$) are trivial.
The case where $t$ is of the form $t_1 \seqc t_2$ follows immediately 
from the induction hypothesis and the claim that, for all head normal 
forms $t_1$ and $t_2$ of \siACP, there exists a head normal form $t'$ of 
\siACP\ such that $t_1 \seqc t_2 = t'$ is derivable from the axioms of 
\siACP.
This claim is easily proved by induction on the structure of~$t_1$.
The case where $t$ is of the form $t_1 \altc t_2$ follows immediately 
from the induction hypothesis.
The cases where $t$ is of one of the forms $t_1 \leftm t_2$, 
$t_1 \commm t_2$, $\encap{H}(t_1)$ or 
$\posm{n}{i}{h}{s}{t_1,\ldots,t_n}$ are proved along the same lines as 
the case where $t$ is of the form $t_1 \seqc t_2$.
In the case that $t$ is of the form $t_1 \commm t_2$, each of the cases 
to be considered in the inductive proof of the claim demands a proof by 
induction on the structure of~$t_2$.
In the case that $t$ is of the form $\posm{n}{i}{h}{s}{t_1,\ldots,t_n}$, 
the claim is of course proved by induction on the structure of $t_i$ 
instead of $t_1$.
The case that $t$ is of the form $t_1 \parc t_2$ follows immediately 
from the case that $t$ is of the form $t_1 \leftm t_2$ and the case that 
$t$ is of the form $t_1 \commm t_2$.
The case that $t$ is of the form $\si{n}{h}{s}{t_1,\ldots,t_n}$ follows 
immediately from the case that $t$ is of the form 
\smash{$\posm{n}{i}{h}{s}{t_1,\ldots,t_n}$}.
Because~$t$ is a guarded \siACP\ term, the case where $t$ is a variable 
cannot occur.

The proof of the proposition itself is also straightforward by induction 
on the structure of $t$.
The cases other than the case where $t$ is of the form $\rec{X}{E}$ is 
proved in the same way as in the above proof of the weaker result.
The case where $t$ is of the form $\rec{X}{E}$ follows immediately from
the weaker result and RDP.
\qed
\end{proof}

Each of the four theorems to come refer to several process algebras.
It is implicit that the same set $\Act$ of actions and the same 
communication function $\commf$ are assumed in the process algebras
referred to.

Each guarded recursive specification over \siACP\ can be reduced to a
guarded recursive specification over \ACP. 
\begin{theorem}[Reduction]
\label{theorem-reduction}
\sloppy
For each guarded recursive specification $E$ over \siACP\ and each 
$X \in \vars(E)$, there exists a guarded recursive specification $E'$ 
over \ACP\ such that $\rec{X}{E} = \rec{X}{E'}$ is derivable from the 
axioms of \siACPr.
\end{theorem}
\begin{proof}
We start with devising an algorithm to construct the guarded recursive 
specification $E'$.
The algorithm keeps a set $V$ of recursion equations from $E'$ that are 
already found and a sequence $W$ of equations of the form 
$X_k = \rec{t_k}{E}$ that still have to be transformed.
The algorithm has a finite or countably infinite number of stages.
In each stage, $V$ and $W$ are finite.
Initially, $V$ is empty and $W$ contains only the equation 
$X_0 = \rec{X}{E}$.

In each stage, we remove the first equation from $W$.
Assume that this equation is $X_k = \rec{t_k}{E}$. 
We bring the term $\rec{t_k}{E}$ into head normal form. 
If $t_k$ is not a guarded term, then we use RDP here to turn $t_k$ into 
a guarded term first.
Thus, by Proposition~\ref{proposition-HNF-siACP}, we can always bring 
the term $\rec{t_k}{E}$ into head normal form.
Assume that the resulting head normal form is
$\Altc{i<n} a_i \seqc t'_i \altc \Altc{j<m} b_j$.
Then, we add the equation 
$X_k = \Altc{i<n} a_i \seqc X_{k+i+1} \altc  \Altc{j<m} b_j$,
where the $X_{k+i+1}$ are fresh variables, to the set $V$.
Moreover, for each $i < n$, we add the equation $X_{k+i+1} = t'_i$ to 
the end of the sequence $W$.
Notice that the terms $t'_i$ are of the form $\rec{t_{k+i+1}}{E}$.

Because $V$ grows monotonically, there exists a limit. 
That limit is the finite or countably infinite linear recursive 
specification $E'$.
Every equation that is added to the finite sequence $W$, is also removed 
from it.
Therefore, the right-hand side of each equation from $E'$ only contains
variables that also occur as the left-hand side of an equation from 
$E'$.

Now, we want to use RSP to show that $\rec{X}{E} = \rec{X}{E'}$ is 
derivable from the axioms of \siACPr.
The variables occurring in $E'$ are $X_0, X_1, X_2, \ldots\;$.
For each $k$, the variable $X_k$ has been exactly once in $W$ as the 
left-hand side of an equation.
For each $k$, assume that  this equation is $X_k = \rec{t_k}{E}$.
To use RSP, we have to show for each $k$ that the equation
$X_k = \Altc{i<n} a_i \seqc X_{k+i+1} \altc  \Altc{j<m} b_j$ from $E'$
with, for each $l$, all occurrences of $X_l$ replaced by $\rec{t_l}{E}$
is derivable from the axioms of \siACPr.
For each $k$, this follows from the construction.
\qed
\end{proof}

The next three theorems will be proved by means of term rewriting 
systems.
In Appendix~\ref{appendix-TRS}, basic definitions and results 
regarding term rewriting systems are collected.
This appendix also serves to fix the terminology on term rewriting 
systems used in the proofs of the next three theorems.

Each closed \siACPrp\ term is derivably equal to a closed \ACPr\ term.
\begin{theorem}[Elimination]
\label{theorem-elimination}
For each closed \siACPrp\ term $t$, there exists a closed \ACPr\ term 
$t'$ such that $t = t'$ is derivable from the axioms of \siACPrp.
\end{theorem}
\begin{proof}
We prove this by means of a term rewriting system that takes equational 
axioms of \siACPrp\ and equations derivable from the axioms of \siACPrp\ 
as rewrite rules.
Thus, the proof boils down to showing that (a)~the term rewriting system 
concerned has the property that each \siACPrp\ term has a unique normal 
form modulo axioms A1 and A2 and (b)~each closed \siACPrp\ term 
that is a normal form modulo axioms A1 and A2 is a closed \ACPr\ term.
Henceforth, we will write AC for the set of equations that consists of 
axioms A1 and~A2.

Let $R$ be a set of equations that contains for each guarded recursive 
specification $E$ over \siACP\ and $X \in \vars(E)$ an equation 
$\rec{X}{E} = \rec{X}{E'}$, where $E'$ is a guarded recursive 
specification over \ACP, that is derivable from the axioms of \siACPrp.
Such a set $R$ exists by Theorem~\ref{theorem-reduction}.
Consider the term rewriting system $\cR(\siACPrp)$ that consists of the 
axioms of \siACPrp, with the exception of A1, A2, RDP, and RSP, and the 
equations from $R$ taken as rewrite rules.

We show that $\cR(\siACPrp)$ has the property that each \siACPrp\ term 
has a unique normal form modulo AC by proving that $\cR(\siACPrp)$ is 
terminating modulo AC and confluent modulo AC.

First, we show that $\cR(\siACPrp)$ is terminating modulo AC. 
This can be proved by the reduction ordering $>$ induced by the extended 
integer polynomials $\theta(t)$ associated with \siACPrp\ terms $t$ as 
follows:%
\footnote
{Here, extended polynomials differ from polynomials in that both
 variables and expressions of the form $2^X$, where $X$ is a variable,
 are allowed where only variables are allowed in polynomials.}
\begin{ldispl}
\begin{geqns}
\theta(X) = \ul{X}\;, \\
\theta(a) = 2\;, \\
\theta(\dead) = 2\;, \\
\theta(\pcr(d)) = \theta(\crea(d))^2 + 1\;, \\
\theta(t_1 \altc t_2) = \theta(t_1) + \theta(t_2)\;, \\
\theta(t_1 \seqc t_2) = \theta(t_1)^2 \cdot \theta(t_2)\;, 
\end{geqns}
\qquad
\begin{geqns}
\theta(t_1 \parc t_2) = 
  3 \cdot (\theta(t_1) \cdot \theta(t_2))^2 + 1\;, \\
\theta(t_1 \leftm t_2) = (\theta(t_1) \cdot \theta(t_2))^2\;, \\
\theta(t_1 \commm t_2) = (\theta(t_1) \cdot \theta(t_2))^2\;, \\
\theta(\encap{H}(t)) = 2^{\theta(t)}\;, \\
\theta(\si{n}{h}{s}{t_1,\ldots,t_n} = 
  (\theta(t_1) \cdot {} \ldots {} \cdot \theta(t_n))^2 + 1\;, \\
\theta(\posm{n}{i}{h}{s}{t_1,\ldots,t_n}) =
  (\theta(t_1) \cdot {} \ldots {} \cdot \theta(t_n))^2\;,
\end{geqns}
\vspace*{1ex} \\ 
\hfill
\begin{geqns}
\theta(\rec{X}{E}) = 
 \left \{
 \begin{array}{l@{\;\;}l}
 2 & \mathrm{if}\; E 
     \mathrm{\;is\;a\;guarded\;recursive\; specification\;over\;} \ACP\\
 3 & \mathrm{otherwise},
 \end{array}
 \right.
\end{geqns}
\hfill
\end{ldispl}%
where it is assumed that, for each variable $X$ over processes, $\ul{X}$ 
is a variable over integers.
The following is easy to see: 
(a)~$t > t'$ for all rewrite rules $t = t'$ of $\cR(\siACPrp)$ and
(b)~$t > t'$ implies $s > s'$ for all \siACPrp\ terms $s$ and $s'$ for 
which $t = s$ and $t' = s'$ are derivable from AC.%
\footnote
{We do not have that $t > t'$ for all rewrite rules $t = s$ if SI2 is 
 replaced by SI2a and SI2b (see Table~\ref{axioms-alt-deadlock}).}
Hence, $\cR(\siACPrp)$ is terminating modulo AC.

Next, we show that $\cR(\siACPrp)$ is confluent modulo AC. 
It follows from Theorems~5 and~16 in~\cite{JK86a} and the fact that 
$\cR(\siACPrp)$ is terminating modulo AC that $\cR(\siACPrp)$ is 
confluent modulo AC if it does not give rise to critical pairs modulo AC
that are not convergent.
It is easy to see that all critical pairs modulo AC arise from 
overlappings of 
(a)~A3 on A4, CM4, CM8, CM9, D3, and SI8,
(b)~A6 on A4, CM4, CM8, CM9, D3, and SI8,
(c)~A7 on CM3, CM5, CM6, CM7, D4, and SI5,
(d)~CM10 on CM9, and (e)~CM11 on CM8.
It is straightforward to check that all critical pairs concerned are
convergent.
Hence, $\cR(\siACPrp)$ is confluent modulo AC.

Above, we have shown that $\cR(\siACPrp)$ is terminating modulo AC and 
confluent modulo AC and by this that it has the property that each 
\siACPrp\ term has a unique normal form modulo AC.
It remains to be shown that each closed \siACPrp\ term that is a normal 
form modulo AC is a closed \ACPr\ term.
It is not hard to see that, for each closed \siACPrp\ term in which other 
operators than $\altc$ and $\seqc$ occur, a reduction step modulo AC is
still possible in $\cR(\siACPrp)$.
Because a reduction step modulo AC is impossible for a normal form 
modulo AC, no other operators than $\altc$ or $\seqc$ can occur in a 
closed \siACPrp\ term that is a normal form modulo AC. 
Hence, each closed \siACPrp\ term that is a normal form modulo AC is a 
closed \ACPr\ term.
\qed
\end{proof}

Each equation between closed \ACP\ terms that is derivable in \siACP\ is 
also derivable in \ACP.
\begin{theorem}[Conservative extension]
\label{theorem-conservativity}
For each two closed \ACP\ terms $t$ and $t'$, $t = t'$ is derivable from 
the axioms of \siACP\ only if $t = t'$ is derivable from the axioms of 
\ACP.
\end{theorem}
\begin{proof}
We prove this by means of a restriction of the term rewriting system 
from the proof of Theorem~\ref{theorem-elimination}.
Consider the term rewriting system $\cR(\siACP)$ that consists of the 
axioms of \siACP, with the exception of A1 and A2.
$\cR(\siACP)$ is $\cR(\siACPrp)$ restricted to \siACP\ terms.
Just like $\cR(\siACPrp)$, $\cR(\siACP)$ is terminating modulo AC and 
confluent modulo AC.
The proofs of these properties for $\cR(\siACPrp)$ carry over to 
$\cR(\siACP)$.

Let $t$ and $t'$ be two closed \ACP\ terms such that $t = t'$ is 
derivable from the axioms of \siACP.
Reduce $t$ and $t'$ to normal forms $s$ and $s'$, respectively, by means 
of the term rewriting system $\cR(\siACP)$.
By Theorem~5 in~\cite{JK86a}, being confluent modulo AC is equivalent 
to being Church-Rosser modulo AC for a term rewriting system that is 
terminating modulo AC.
This means that $t$ and $t'$ have the same normal form modulo AC.
In other words, $s = s'$ is derivable from axioms A1 and A2.
Because (a)~no other operators than $\altc$ and $\seqc$ occur in $t$ and 
$t'$ and (b)~no rewrite rule introduces one or more of the other 
operators if one or more of the other operators was not already in its 
left-hand side, each rewrite rule applied in the reduction from $t$ to 
$s$ or the reduction from $t'$ to $s'$ is one of the axioms of \ACP.
Therefore, each rewrite rule involved in the reduction from $t$ to $s$ 
or the reduction from $t'$ to $s'$ is an axiom of \ACP.
Hence, the reduction from $t$ to $s$ shows that $t = s$ is derivable 
from the axioms of \ACP\ and the reduction from $t'$ to $s'$ shows that
$t' = s'$ is derivable from the axioms of \ACP.
From this and the fact that $s = s'$ is derivable from axioms A1 and A2,
it follows $t = t'$ is derivable from the axioms of \ACP. 
\qed
\end{proof}

The following theorem concerns the expansion of minimal models of \ACP\ 
to models of \siACP.
\begin{theorem}[Unique expansion]
\label{theorem-expansion}
\sloppy
Each minimal model of \ACP\ has a unique expansion to a model of \siACP.
\end{theorem}
\begin{proof}
We write $f^\cA$, where $\cA$ is a model of \ACP\ or \siACP\ and $f$ is 
a constant or operator from the signature of $\cA$, for the 
interpretation of $f$ in $\cA$.
We write $t^\cA$, where $\cA$ is a model of \ACP\ or \siACP\ and $t$ is 
a closed term over the signature of $\cA$, for the interpretation of 
$t$ in $\cA$.

Let $\cA$ be a minimal model of \ACP.
Let $\cterm$ be a function from the carrier of $\cA$ to the set of all 
closed \ACP\ terms such that, for each element $p$ of the carrier of 
$\cA$, $\cterm(p)^\cA = p$.
Because $\cA$ is a minimal model of \ACP, $\cterm(p)$ is a total 
function.
We write $\ul{p}$, where $p$ is an element of the carrier of $\cA$, for
$\cterm(p)$.
Let $\elim$ be a function from the set of all closed \siACP\ terms to 
the set of all closed \ACP\ terms such that, for each closed \siACP\ 
term $t$, $\elim(t)$ is one of the normal forms that $t$ can be reduced 
to by means of the term rewriting system $\cR(\siACP)$ from the proof 
of Theorem~\ref{theorem-conservativity}.

We start with constructing an expansion of $\cA$ with interpretations of 
the additional operators of \siACP.
Let $\cB$ be the expansion of $\cA$ with interpretations of the 
additional operators of \siACP\ where these interpretations are defined 
as follows:
\begin{ldispl}
\siop{n}{h}{s}^\cB(p_1,\ldots,p_n) =
{\elim(\siop{n}{h}{s}(\ul{p_1},\ldots,\ul{p_n}))}^\cA\;,
\\
\posmop{n}{i}{h}{s}^{\cB}(p_1,\ldots,p_n) =
{\elim(\posmop{n}{i}{h}{s}(\ul{p_1},\ldots,\ul{p_n}))}^\cA\;,
\end{ldispl}%
for all $p_1,\ldots,p_n$ from the carrier of $\cA$.

We proceed with proving that $\cB$ is a model of \siACP.
By Theorem~\ref{theorem-conservativity}, it is sufficient to prove that
$\cB$ satisfies axioms SI1--SI8.
By its construction, $\cB$ is a minimal algebra and consequently it is 
sufficient to prove that $\cB$ satisfies all closed substitution 
instances of SI1--SI8.
We use the following three claims to prove this:
\begin{itemize}
\item
for all closed substitution instances $t = t'$ of SI1--SI8, 
$t^\cB = {\elim(t)}^\cA$;
\item
for all closed substitution instances $t = t'$ of SI1--SI8, 
$t'^\cB = {\elim(t')}^\cA$;
\item
for all closed substitution instances $t = t'$ of SI1--SI8, 
$\elim(t)^\cA = \elim(t')^\cA$.
\end{itemize}
The first claim follows easily from the definitions of the 
interpretations of the additional operators of \siACP\ given above.
The second claim follows easily from these definitions and the proof of
the first claim. 
Because $\cR(\siACP)$ is Church-Rosser modulo AC (see the proof of 
Theorem~\ref{theorem-conservativity}), we have that 
$\elim(t) = \elim(t')$ is derivable from axioms A1 and A2.
From this, the third claim follows immediately.
It is an immediate consequence of the three claims that $\cB$ satisfies 
all closed substitution instances of SI1--SI8.

We still have to prove that $\cB$ is the only expansion of $\cA$ to a 
model of \siACP.
We can prove this by contradiction. 
Assume that $\cC$ is an expansion of $\cA$ to a model of \siACP\ that
differs from $\cB$.
Then at least one of the additional operators of \siACP\ has different
interpretations in $\cB$ and $\cC$. 
By the definitions of the interpretations of the additional operators of
\siACP\ in $\cB$, this means that there exists a closed \siACP\ term 
$t$ such that $t^\cC \neq \elim(t)^\cA$. 
Moreover, because because $t = \elim(t)$ is derivable from the axioms 
of \siACP, $t^\cC = \elim(t)^\cC$.
Hence, $\elim(t)^\cC \neq \elim(t)^\cA$.
Because $\elim(t)$ is a closed \ACP\ term, this contradicts the fact 
that $\cC$ is an expansion of $\cA$.
\qed
\end{proof}

\section{Concluding Remarks}
\label{sect-concl}

We have extended the algebraic theory of processes known as \ACP\ with
the form of interleaving that underlies multi-threading as found in 
contemporary programming languages.
We have also established some basic properties of the resulting theory.
It remains an open question whether strategic interleaving is definable
in an established extension of \ACP.

In~\cite{Mid20a}, we extend a minor variant of \ACP, known as 
\ACP$_\epsilon$, with strategic interleaving and show that the resulting 
theory can deal with a process-scheduling policy that supports mutual 
exclusion of critical subprocesses.

\appendix

\section{Structural Operational Semantics of ACP+SI}
\label{appendix-SOS}

It is customary to associate transition systems with closed terms of the 
language of an ACP-like theory about processes by means of structural 
operational semantics and to use this to construct a model in which 
closed terms are identified if their associated transition systems are 
bisimilar.
The structural operational semantics of \ACP\ can be found 
in~\cite{BW90,Fok00}.
The additional transition rules for the strategic interleaving operators 
and the positional strategic interleaving operators are given in 
Table~\ref{trules-SI}.
\begin{table}[!t]
\caption{Transition rules for strategic interleaving}
\label{trules-SI}
\begin{ruletbl}
\Rule
{\aterm{x}{a}}
{\aterm{\si{1}{h}{s}{x}}{a}}
\\
\Rule
{\aterm{x_i}{a} \qquad i = \sched{n}(h,s)}
{\astep{\si{n+1}{h}{s}{x_1,\ldots,x_{n+1}}}{a}
 {\si{n}{h \concat \tup{i,n}}{\updat{n+1}(h,s,i,a)}
  {x_1,\ldots,x_{i-1},x_{i+1},\ldots,x_{n+1}}}}
\\
\Rule
{\astep{x_i}{a}{x'_i} \qquad i = \sched{n}(h,s)}
{\astep{\si{n}{h}{s}{x_1,\ldots,x_n}}{a}
 {\si{n}{h \concat \tup{i,n}}{\updat{n}(h,s,i,a)}
  {x_1,\ldots,x_{i-1},x_i',x_{i+1},\ldots,x_n}}}
\\
\Rule
{\aterm{x_i}{\pcr(d)} \qquad i = \sched{n}(h,s)}
{\astep{\si{n}{h}{s}{x_1,\ldots,x_n}}{\pcr(d)}
 {\si{n}{h \concat \tup{i,n}}{\updat{n}(h,s,i,\pcr(d))}
   {x_1,\ldots,x_{i-1},x_{i+1},\ldots,x_n,\crea(d)}}}
\\
\Rule
{\astep{x_i}{\pcr(d)}{x'_i} \qquad i = \sched{n}(h,s)}
{\astep{\si{n}{h}{s}{x_1,\ldots,x_n}}{\pcr(d)}
 {\si{n+1}{h \concat \tup{i,n+1}}{\updat{n}(h,s,i,\pcr(d))}
 {x_1,\ldots,x_{i-1},x_i',x_{i+1},\ldots,x_n,\crea(d)}}}
\\
\Rule
{\aterm{x}{a}}
{\aterm{\posm{1}{i}{h}{s}{x}}{a}}
\\
\Rule
{\aterm{x_i}{a}}
{\astep{\posm{n+1}{i}{h}{s}{x_1,\ldots,x_{n+1}}}{a}
 {\si{n}{h \concat \tup{i,n}}{\updat{n+1}(h,s,i,a)}
  {x_1,\ldots,x_{i-1},x_{i+1},\ldots,x_{n+1}}}}
\\
\Rule
{\astep{x_i}{a}{x'_i}}
{\astep{\posm{n}{i}{h}{s}{x_1,\ldots,x_n}}{a}
 {\si{n}{h \concat \tup{i,n}}{\updat{n}(h,s,i,a)}
  {x_1,\ldots,x_{i-1},x_i',x_{i+1},\ldots,x_n}}}
\\
\Rule
{\aterm{x_i}{\pcr(d)}}
{\astep{\posm{n}{i}{h}{s}{x_1,\ldots,x_n}}{\pcr(d)}
 {\si{n}{h \concat \tup{i,n}}{\updat{n}(h,s,i,\pcr(d))}
   {x_1,\ldots,x_{i-1},x_{i+1},\ldots,x_n,\crea(d)}}}
\\
\Rule
{\astep{x_i}{\pcr(d)}{x'_i}}
{\astep{\posm{n}{i}{h}{s}{x_1,\ldots,x_n}}{\pcr(d)}
 {\si{n+1}{h \concat \tup{i,n+1}}{\updat{n}(h,s,i,\pcr(d))}
 {x_1,\ldots,x_{i-1},x_i',x_{i+1},\ldots,x_n,\crea(d)}}}
\end{ruletbl}
\end{table}
In this table,
\begin{itemize}
\item
$\aterm{t}{a}$ indicates that $t$ is capable of performing action $a$ 
and then terminating successfully;
\item
$\astep{t}{a}{t'}$ indicates that $t$ is capable of performing action 
$a$ and then behaving as $t'$.
\end{itemize}
The transition rules for the strategic interleaving operator are similar 
to the transition rules for the positional strategic interleaving 
operators, but each transition rule for the strategic interleaving operator 
has the side-condition $i = \sched{n}(h,s)$.

\section{Term Rewriting Systems}
\label{appendix-TRS}

In this appendix, basic definitions and results regarding term 
rewriting systems are collected.
This appendix also serves to fix the terminology on term rewriting 
systems used in the proofs that make use of term rewriting systems.

We assume that a set of constants, a set of operators with fixed 
arities, and a set of variables have been given; and we consider an 
arbitrary term rewriting system $\cR$ for terms that can be built from 
the constants, operators, and variables in these sets.

A \emph{rewrite rule} is a pair of terms $t \osred s$, where $t$ is not 
a variable and each variable occurring in $s$ occurs in $t$ as well.
A \emph{term rewriting system} is a set of rewrite rules.

A \emph{reduction step} of $\cR$ is a pair $t \osred s$ such that for 
some substitution instance $t' \osred s'$ of a rewrite rule of $\cR$, 
$t'$ is a subterm of $t$, and $s$ is $t$ with $t'$ replaced by $s'$.
Here, $t'$ is called the \emph{redex} of the reduction step.
A \emph{reduction} of $\cR$ is a pair $t \msred s$ such that either 
$t \equiv s$ or there exists a finite sequence $t_1 \osred t_2$, \ldots, 
$t_n \osred t_{n+1}$ of consecutive reduction steps of $\cR$ such that 
$t \equiv t_1$ and $s \equiv t_{n+1}$.

A term $t$ \emph{is a normal form of $\cR$} if there does not exist a 
term $s$ such that $t \osred s$ is a reduction step of $\cR$.
A term $t$ \emph{has a normal form in $\cR$} if there exists a reduction 
$t \msred s$  of $\cR$ and $s$ is a normal form of $\cR$.
$\cR$ \emph{is terminating} on term $t$ if there does not exist an 
infinite sequence $t \osred t_1$, $t_1 \osred t_2$, $t_2 \osred t_3$, 
\ldots\ of consecutive reduction steps of $\cR$.
$\cR$ \emph{is terminating} if $\cR$ is terminating on all terms.
$\cR$ \emph{is confluent} if for all reductions $t \msred s_1$ and 
$t \msred s_2$ of $\cR$ there exist reductions $s_1 \msred s$ and 
$s_2 \msred s$ of $\cR$.
If $\cR$ is terminating and confluent, then each term has a unique 
normal form in $\cR$.

A \emph{reduction ordering} for $\cR$ is a well-founded ordering on 
terms that is closed under substitutions and contexts.
$\cR$ is terminating if and only if there exists a reduction ordering 
$>$ for $\cR$ such that $t > s$ for each rewrite rule $t \osred s$ of 
$\cR$.

A \emph{unifier} of two terms $s$ and $t$ is a substitution $\sigma$ 
such that $\sigma(s) \equiv \sigma(t)$.
A \emph{critical pair} of $\cR$ is a pair $(t_1,t_2)$ of terms for which 
there exist rewrite rules $s \osred s'$ and $t \osred t'$ of $\cR$ and a
`most general unifier' $\sigma$ of $s$ and a non-variable subterm of $t$ 
such that $t_1 \equiv \sigma(t'')$ and $t_2 \equiv \sigma(t')$, where 
$t''$ is $t$ with $\sigma(s)$ replaced by $\sigma(s')$.%
\footnote
{See e.g.\ Definition~10 in~\cite{JK86a} for the definitions of most 
general unifier and complete set of unifiers modulo $E$.}%
\addtocounter{footnote}{-1}
A critical pair $(t_1,t_2)$ of $\cR$ \emph{is convergent} if there exist 
reductions $t_1 \msred s$ and $t_2 \msred s$ of $\cR$.
If $\cR$ is terminating, then $\cR$ is confluent if and only if all 
critical pairs of $\cR$ are convergent.

Henceforth, we consider an arbitrary set $E$ of equations between terms.

A \emph{reduction step modulo $E$} of $\cR$ is a pair $t \osredm{E} s$ 
such that there exists a reduction step $t' \osred s'$ of $\cR$ such 
that $t = t'$ and $s = s'$ are derivable from $E$.
A \emph{reduction modulo $E$} of $\cR$ is pair $t \msredm{E} s$ such 
that either $t = s$ is derivable from $E$ or there exists a finite 
sequence $t_1 \osredm{E} t_2$, \ldots, $t_n \osredm{E} t_{n+1}$ of 
consecutive reduction steps modulo $E$ of $\cR$ such that $t \equiv t_1$ 
and $s \equiv t_{n+1}$.

A term $t$ \emph{is a normal form modulo $E$ of $\cR$} if there does not 
exist a term $s$ such that $t \osredm{E} s$ is a reduction step modulo 
$E$ of $\cR$.
A term $t$ \emph{has a normal form modulo $E$ in $\cR$} if there exists 
a reduction modulo $E$ $t \msredm{E} s$ of $\cR$ and $s$ is a normal 
form modulo $E$ of $\cR$.
$\cR$ \emph{is terminating modulo $E$} on term $t$ if there does not 
exist an infinite sequence $t \osredm{E} t_1$, $t_1 \osredm{E} t_2$, 
$t_2 \osredm{E} t_3$,~\ldots\ of consecutive reduction steps modulo $E$ 
of $\cR$.
$\cR$ \emph{is terminating modulo $E$} if $\cR$ is terminating modulo 
$E$ on all terms.
$\cR$ \emph{is confluent modulo $E$} if for all reductions modulo 
$E$ $t \msredm{E} s_1$ and $t \msredm{E} s_2$ of $\cR$ there exist 
reductions modulo $E$ $s_1 \msredm{E} s$ and $s_2 \msredm{E} s$ of 
$\cR$.
If $\cR$ is terminating modulo $E$ and confluent modulo $E$, then each 
term has a unique normal form modulo $E$ in $\cR$.

A reduction ordering $>$ for $\cR$ \emph{is $E$-compatible} if $t > s$
implies $t' > s'$ for all terms $t'$ and $s'$ for which 
$t = t'$ and $s = s'$ are derivable from $E$.
$\cR$ is terminating modulo $E$ if and only if there exists an 
$E$-compatible reduction ordering $>$ for $\cR$ such that $t > s$ for 
each rewrite rule $t \osred s$ of $\cR$.

A \emph{unifier modulo $E$} of two terms $s$ and $t$ is a substitution 
$\sigma$ such that $\sigma(s) = \sigma(t)$ is derivable from $E$.
A \emph{critical pair modulo $E$} of $\cR$ is a pair $(t_1,t_2)$ of 
terms for which there exist rewrite rules $s \osred s'$ and 
$t \osred t'$ of $\cR$ and a substitution $\sigma$ from a `complete set 
of unifiers modulo $E$' of $s$ and a non-variable subterm of $t$ such 
that $t_1 \equiv \sigma(t'')$ and $t_2 \equiv \sigma(t')$, where $t''$ 
is $t$ with $\sigma(s)$ replaced by $\sigma(s')$.\footnotemark\
If $\cR$ is terminating modulo $E$, then $\cR$ is confluent modulo $E$ 
if and only if all critical pairs modulo $E$ of $\cR$ are convergent.

An \emph{$E$-equality step} is a pair $t \oseq{E} s$ such that, for some 
substitution instance $t' = s'$ of an equation from $E$, either $t'$ is 
a subterm of $t$ and $s$ is $t$ with $t'$ replaced by $s'$ or $s'$ is a 
subterm of $t$ and $s$ is $t$ with $s'$ replaced by $t'$.
An \emph{$\cR \union E$-equality step} is a pair $t \osconeq{E} s$ such 
that $t \osred s$ is a reduction step of $\cR$ or $s \osred t$ is a 
reduction step of $\cR$ or $t \oseq{E} s$ is an $E$-equality step.
An \emph{$\cR \union E$-equality} is a pair $t \msconeq{E} s$ such that 
either $t \equiv s$ or there exists a finite sequence 
$t_1 \osconeq{E} t_2$, \ldots, $t_n \osconeq{E} t_{n+1}$ of consecutive
$\cR \union E$-equality steps such that $t \equiv t_1$ and 
$s \equiv t_{n+1}$.
$\cR$ is \emph{Church-Rosser modulo $E$} if for all 
$\cR \union E$-equalities $t \msconeq{E} t'$ there exist reductions 
modulo $E$ $t \msredm{E} s$ and $t' \msredm{E} s$ of $\cR$.
If $\cR$ is terminating modulo $E$, then $\cR$ is Church-Rosser modulo 
$E$ if and only if $\cR$ is confluent modulo $E$.

\bibliographystyle{splncs03}
\bibliography{PA}

\end{document}